\documentclass{patmorin}
\usepackage{graphicx,amsmath}
\usepackage{pat}
\usepackage[mathlines]{lineno}
\newcommand{\eps}{\epsilon}
\newcommand{\depth}{\mathrm{depth}}
\newcommand{\poly}{\includegraphics{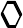}}

\title{\MakeUppercase{Odds-On Trees}}

\author{Prosenjit~Bose, 
        Luc~Devroye,
	Karim~Dou\"{\i}eb, 
	Vida~Dujmovi\'c, 
	James~King, and 
	Pat~Morin}

\begin{document}
\maketitle

\begin{abstract}
  Let $\mathcal{P}:\R^d\rightarrow A$ be a query problem over $\R^d$
  for which there exists a data structure $\mathcal{S}$ that can compute
  $\mathcal{P}(q)$ in $O(\log n)$ time for any query point $q\in\R^d$.
  Let $D$ be a probability measure over $\R^d$ representing a distribution
  of queries.  We describe a data structure $T=T_{\mathcal{P},D}$, called
  the \emph{odds-on tree}, of size $O(n^\eps)$ that can be used as a
  filter that quickly computes $\mathcal{P}(q)$ for some query values in
  $\R^d$ and relies on $\mathcal{S}$ for the remaining queries.  With an
  odds-on tree, the expected query time for a point drawn according to
  $D$ is $O(H^*+1)$, where $H^*$ is a lower-bound on the expected cost
  of any linear decision tree that solves $\mathcal{P}$.

  Odds-on trees have a number of applications, including
  distribution-sensitive data structures for point location in 2-d,
  point-in-polytope testing in $d$ dimensions, ray shooting in simple
  polygons, ray shooting in polytopes, nearest-neighbour queries in
  $\R^d$, point-location in arrangements of hyperplanes in $\R^d$,
  and many other geometric searching problems that can be solved in the
  linear-decision tree model.  A standard lifting technique extends these
  results to algebraic decision trees of constant degree.  A slightly
  different version of odds-on trees yields similar results for orthogonal
  searching problems that can be solved in the comparison tree model.
\end{abstract}

\section{Introduction}

Geometric search problems have a long and rich history
\cite{ae99,g00chapter,ms05}. In these problems, we are typically
given a collection $S$ of $n$ geometric objects in $\R^d$ and asked
to store them in a data structure so that we can efficiently answer
queries about these objects. Until quite recently, the performance of
these data structures was typically measured, at least theoretically,
in terms of the worst-case (over all possible choices of $S$ and $q$)
query time as a function of $n$.  A worst-case query time of $O(\log
n)$ is typically viewed as the gold standard for such search problems,
at least in models of computation that only allow binary decisions.

Somewhat more recently, researchers have begun studying geometric
search problems under the lens of \emph{distribution-sensitivity}. In
this setting, one assumes that there is a probability measure $D$ over
the set of possible queries, and one attempts to optimize the expected
query time when queries are distributed according to $D$.  For example,
given a planar triangulation $G$ and a distribution $D$ over $\R^2$,
one can construct an $O(n)$ sized data structure that can determine
the face of $G$ that contains any query point $q\in\R^2$.  The expected
query time of this structure is $O(H+1)$, where
\begin{equation}
    H = \sum_{i} p_i\log (1/p_i)  \enspace  \eqlabel{entropy-faces}
\end{equation}
and $p_i$ denotes the probability that $q$ is contained in the $i$th face
of $G$ \cite{acmr00,amm00,amm01a,amm01b,ammw07,i01,i04}.  Information
theory tells us that this result is optimal in any \emph{dichotomous}
model of computation where any execution branching has at most 2 possible
outcomes.

More recently, Collette \etal\ \cite{cdilm08,cdilm09} have shown that a
similar result holds for point location in any simple connected planar
subdivision $G$.  Note that this involves more than simply triangulating
$G$ and applying the results for triangulations.  Triangulating $G$
increases the number of faces and therefore also increases the value of
$H$ in \eqref{entropy-faces}.  On the other hand, we can not expect to
always achieve a query time of $O(H+1)$ since this would imply, for
example, an $O(1)$ query-time data structure for testing if a point
is contained in a simple polygon.  The result of Collette \etal\  is
a data structure whose expected query time is $O(H^*+1)$ where $H^*$
is a lower-bound on the expected cost of any linear decision tree for
point location in $G$ with queries distributed according to $D$.

Dujmovi\'c \etal\ \cite{dhm09} study distribution-sensitive 2-sided
2-dimensional orthogonal range counting.  They describe a data
structure, called a biased range tree, that preprocesses an $n$ point
set $S\subset\R^2$ and a query distribution $D$ over $\R^2$.  A biased
range tree uses $O(n\log n)$ preprocessing time and space and can
answer 2-sided 2-dimensional orthogonal range counting queries over
$S$ in expected time $O(H^*+1)$, where $H^*$ is a lower-bound on the
expected cost of any comparison tree for 2-sided range queries over $S$
with queries distributed according to $D$.

The results of Collette \etal\ \cite{cdilm08,cdilm09} and Dujmovi\'c
\etal\ \cite{dhm09} use similar techniques to prove their optimality,
but each requires their own \textit{ad hoc} arguments.  For example, the point
location results are achieved by finding near-minimum-entropy Steiner
triangulations and then applying existing distribution-sensitive results
for point location in triangulations.  Biased range trees on the other
hand, mix $k$-d trees, fractional cascading, and biased binary search
trees to achieve their running time.

In the current paper, we describe a general and low-overhead method
of taking any $O(\log n)$ query-time data structure and making it
distribution-sensitive.  Let $\mathcal{P}:\R^d\rightarrow A$ be a query
problem over $\R^d$ for which there exists a data structure $\mathcal{S}$
that can compute $\mathcal{P}(q)$ in $O(\log n)$ time for any query
point $q\in\R^d$.\footnote{There is no natural definition of $n$ for
the abstract problem $\mathcal{P}$.  Nevertheless, all the problems
we will eventually consider arise from a set of $n$ objects in $\R^d$
and, for all these problems $O(\log n)$ is the optimal worst-case
query time in the models of computation we will consider.}  Let $D$
be a probability measure over $\R^d$ representing a distribution
of queries.  We describe a data structure $T=T_{\mathcal{P},D}$,
called the \emph{odds-on tree},\footnote{\textbf{odds-on} $\cdot$
\textit{adj.}\ having a better than even chance of success ``the odds-on
favourite''}  of size $O(n^\eps)$ that can be used as a filter that
quickly computes $\mathcal{P}(q)$ for some query values in $\R^d$ and
relies on $\mathcal{S}$ for the remaining queries.  With an odds-on tree,
the expected query time for a point drawn according to $D$ is $O(H^*+1)$,
where $H^*$ is a lower-bound on the expected cost of any linear decision
tree that solves $\mathcal{P}$.

For any constant integer $p\ge 1$, a standard lifting technique allows us
to lift queries from $\R^d$ into $\R^{d'}$, with $d'={p+d \choose p}-1$
so that there is a correspondence between $d'$-variate linear inequalities
(halfspaces in $\R^{d'}$) and $d$-variate polynomial inequalities of
maximum degree $p$ \cite{yy85}.  This technique yields the same result,
except that $H^*$ becomes a lower-bound on the expected cost of any
degree $p$ algebraic decision tree that solves $\mathcal{P}$.

Odds-on trees have a plethora of applications, including
distribution-sensitive data structures for point location in 2-d,
point-in-polytope testing in $d$ dimensions, ray shooting in simple
polygons, ray shooting in polytopes, nearest-neighbour queries in $\R^d$,
point-location in arrangements of hyperplanes in $\R^d$, and many other
geometric searching problems that can be solved in the linear-decision
tree model.  Furthermore, a variant of the odds-on tree that works in
the comparison tree model provides distribution-sensitive data structures
for orthogonal searching problems in the comparison tree model.

The remainder of this paper is organized as follows:
\Secref{prelim} presents some preliminary definitions and background
material. \Secref{data-structure} presents the odds-on tree and
algorithms for constructing it.  \Secref{analysis} proves that the
odds-on tree matches the query time of any linear decision tree.
\Secref{applications} presents some of the geometric applications of this
data structure. Finally, \secref{conclusions} summarizes and concludes
with directions for future work.

\section{Preliminaries}
\seclabel{prelim}

Throughout this paper, the underlying dimension, $d$, is a constant, and
other constants defined in the paper may (implicitly) depend on $d$.
A \emph{simplex} in $\R^d$ is the common intersection of a set of at most
$d+1$ closed halfspaces in $\R^d$. Note that, under this definition,
simplices need not be bounded and $\R^d$, as well as $\emptyset$, are
both simplices.

Throughout this paper, we assume an underlying probability measure
$D$ over $\R^d$.  All expectations and probabilities are (implicitly)
with respect to $D$.  For any subset $X\subseteq\R^d$, $\Pr(X)$ refers
to $D(X)$.  We use the notation $D_{|X}$ to denote the distribution $D$
conditioned on $X$, i.e., $D_{|X}(Y)=\Pr(Y\mid X)=\Pr(X\cap Y)/\Pr(X)$
for all $Y\subseteq\R^d$.  If $F=\{X_1,\ldots,X_k\}$ are non-overlapping
subsets of $\R^d$ then the \emph{entropy} of $F$, denoted $H(F)$ is
\[
    H(F) = \sum_{i=1}^k \Pr(X_i|{\cup F})\log(1/\Pr(X_i|{\cup F})) \enspace ,
\]
where for any set $S$, $\cup S$ denotes $\bigcup_{s\in S} s$.
The probability measure $D$ is used as an input to our algorithms.
We assume that the algorithm has access to $D$ through a \emph{Sampling
Oracle} that allows us to draw a random sample $q\in\R^d$ from the
distribution $D$.

A \emph{query problem} over $\R^d$ is a function
$\mathcal{P}:\R^d\rightarrow A$ where $A$ is some set of \emph{answers}.
We assume the existence of two oracles that allow access to $\mathcal{P}$.
The \emph{Backup Oracle} allows us to compute $\mathcal{P}(q)$ for
any query point $q$, but requires $O(\log n)$ time to do so.  The
\emph{Interference Oracle} allows us to test, for any simplex $\Delta$,
if there exists $p,q\in\Delta$ with $\mathcal{P}(p)\neq\mathcal{P}(q)$.
The running time of the Interference Oracle will be unspecified.

The oracles are used in the following ways:  The Backup Oracle is used to
answer any queries that can not be answered directly by the odds-on tree.
The Sampling Oracle and the Interference Oracle are used only during
the construction of the odds-on tree. The running times in theorems in
Sections~\ref{sec:data-structure} and \ref{sec:analysis} all specify
the number of invocations of the Sampling and Interference Oracle
used. When discussing applications in \secref{applications}, efficient
implementations of the Interference Oracle for specific problems will
be described.

A \emph{decision tree} for $\mathcal{P}$ is a rooted ordered binary
tree in which each internal node $v$ is labelled with a function
$v_f:\R^d\rightarrow\{0,1\}$ and each leaf $w$ is labelled with an element
$w_A\in A$.  A query point $q\in\R^d$ follows a root-to-leaf path,
proceeding to the right child of $v$ if $v_f(q)=1$ and the left
child of $v$ if $v_f(q)=0$.  For a decision tree $T$ and a point
$q\in\R^d$, we denote by $T(q)$ the label of the leaf on the root-to-leaf
path for $p$ in $T$.  A decision tree \emph{solves} $\mathcal{P}$ if
$T(q)=\mathcal{P}(q)$ for all $q\in\R^d$. The \emph{depth} of a node $v$
in a rooted tree $T$, denoted $\depth_T(v)$, is the number of edges
on the path from $v$ to the root of $T$.  The \emph{(expected) cost}
of a decision tree, denoted $\mu_D(T)$, is the expected depth of the
leaf reached when $p$ is drawn according to the probability measure $D$.

Decision trees are classified based on the types of functions, $v_f$,
used at their nodes.  In a \emph{linear decision tree}, each $v_f$ is
a linear inequality.  In a \emph{degree $p$ algebraic decision tree},
each $v_f$ is a $d$-variate polynomial inequality of degree $p$.  In a
\emph{comparison tree}, $v_f$ is a simple comparison that compares one
coordinate of a query point $q$ to some value.

Entropy, query problems, and decision trees are all related by (half of)
Shannon's Source Coding Theorem \cite{s48}:

\begin{thm}[Shannon 1948]\thmlabel{shannon}
  Let $\mathcal{P}:\R^d\rightarrow A$ be a query problem and define
  $\mathcal{P}^{-1}(a)=\{q\in\R^d:\mathcal{P}(q)=a\}$, for any $a\in
  A$. Then, for any decision tree $T$ that solves $\mathcal{P}$,
  \[
    \mu_D(T) \ge H(\{\mathcal{P}^{-1}(a): a\in A\}) \enspace .
  \]
\end{thm}

\section{The Data Structure}
\seclabel{data-structure}

In this section we describe a data structure $T_{\mathcal{P},D}$ called
the \emph{odds-on tree} that, in conjunction with a Backup Oracle, yields
a data structure that solves $\mathcal{P}$ and has expected query time
that is within a constant factor of the expected cost of any linear
decision tree $T^*$ for $\mathcal{P}$.   For a reader familiar with
Matou\v{s}ek's efficient partition trees \cite{m92},  the executive
summary of this section is as follows: An odds-on tree is essentially
a partition tree on a sample of $O(n^\epsilon)$ points drawn according
to the distribution $D$.  A complete description follows.

\begin{thm}[Matou\v{s}ek 1992]\thmlabel{point-partition}
There exists a constant $c$ such that, for any set $S$ of $m$
points in $\R^d$ and any constant $r\le m$, there exists a sequence
$\langle \Delta_1,\ldots,\Delta_r\rangle$ of closed simplices such that
$\bigcup_{i=1}^r \Delta_i = \R^d$,
  \begin{enumerate}
    \item $|\Delta_i^*\cap S| \le 2m/r$, where $\Delta_i^*=\Delta_i\setminus
    \left(\bigcup_{j=1}^{i-1}\Delta_j\right)$, and
    \item For any hyperplane $\ell$, there are at most $cr^{1-1/d}$ elements of
  $\{\Delta_1,\ldots,\Delta_r\}$ whose interiors intersect $\ell$.
  \end{enumerate}
  The sequence of simplices $\Delta_1,\ldots,\Delta_r$ can be computed
  in $O(m)$ time.
\end{thm}

Note that Condition~1 of \thmref{point-partition} is not in the original
statement of the theorem, but follows from Matou\v{s}ek's incremental
construction of $\Delta_1,\ldots,\Delta_r$ \cite{m92}.

Let $S$ be a set of $m$ points in $\R^d$. Then the \emph{partition
tree} $T_S$ for $S$ is a rooted ordered tree obtained by recursively
applying \thmref{point-partition}.  The root of $T_S$ has $r$ children
corresponding to the simplices $\Delta_1,\ldots,\Delta_{r}$ obtained
by applying \thmref{point-partition} to $S$. The $i$th child of the
root is itself the root of the partition tree $T_{S\cap\Delta_i^*}$
for $S\cap\Delta_i^*$. This recursive process stops when the set $S$
contains at most 1 point or when the depth exceeds some pre-specified
maximum depth $k$.

Next we define some regions, $\Delta(v)$, $\poly(v)$, and $\Xi(v)$,
that are associated with each node $v$ of $T_S$.  Every node $v$ in
$T_S$, except the root of $T_S$, is naturally associated with a simplex
$\Delta(v)$ that was obtained from \thmref{point-partition} and that
generated $v$.  For the root of $T_S$, we define $\Delta(v)=\R^d$.

For a node $v$ of $T$ whose ancestors are $v_1,\ldots,v_i$
we define $\poly(v)=\Delta(v)\cap\bigcap_{j=1}^i \Delta(v_i)$.  Note that
$\poly(v)\subseteq\Delta(v)$, and that $\poly(v)$ is a convex polytope
that has $O(((d+1)(i+1))^{\floor{d/2}})=O(i^{\floor{d/2}})$ vertices since
it is the intersection of at most $(d+1)(i+1)$ halfspaces \cite{m70}.

For a point $q\in\R^d$, the \emph{search path} for $q$ in $T_S$ starts
at the root and proceeds to the first child $i$ such that $q\in\Delta_i$
(note that this implies $q\in\Delta_i^*$) and this process is applied
recursively until reaching a leaf of $T_S$.  In this way, for every node
$v$ of the partition tree there is a maximal subset $\Xi(v)\subseteq
\R^d$ such that the search path for every point $q\in\Xi(v)$ contains $v$.
Note that $\Xi(v)\subseteq \poly(v)\subseteq \Delta(v)$, but that $\Xi(v)$
is not necessarily convex or even connected.

We extend the definitions of $\Xi$, $\poly$, and $\Delta$ to sets of
nodes in $T$ in the natural way; if $V$ is a set of nodes in $T$,
then $\Xi(V)=\{\Xi(v):v\in V\}$, $\poly(V)=\{\poly(v):v\in V\}$, and
$\Delta(V)=\{\Delta(v):v\in V\}$.

The following theorem summarizes the properties of the partition tree
$T_{S}$ \cite{m92} (each property is inherited from the corresponding
property of the simplicial partition in
\thmref{point-partition}):

\begin{thm}\thmlabel{point-partition-tree}
  Let $S$ be a set of $m$ points in $\R^d$, let $T_S$ denote the partition
  tree described above, and let $V_i$ denote the set of at most $r^i$
  nodes of $T_S$ at depth $i$.  There exists a constant $c$, independent
  of $r$ and $m$, such that the partition tree $T_S$ has the following
  properties, for every $i\in\Z$:
  \begin{enumerate}
    \item For every node $v\in V_i$, $|S\cap\Xi(v)| \le m(2/r)^i$, and 
    \item For any hyperplane $\ell$, the number of elements in
      $\poly(V_i)$ whose interiors intersect $\ell$ is at most
      $(cr^{1-1/d})^i$.
    \end{enumerate}
    The partition tree $T_S$ can be constructed from $S$ in
    $O(m\log m)$ time.
\end{thm}

The following is a sampling version of \thmref{point-partition-tree} that we
will use in the construction of an odds-on tree:

\begin{thm}\thmlabel{prob-partition-tree} 
  Let $S$ be a sample of $m$ points in $\R^d$ i.i.d.\ according
  to $D$, let $T_D=T_S$ denote the partition tree
  given by \thmref{point-partition-tree}, and let $V_i$ denote the set
  of at most $r^i$ nodes of $T_D$ at depth $i$.  There exists a constant $c$,
  independent of $r$, $m$ and $D$, such that with probability at least
  $1-O(e^{-m^{1/2}})$, $T_D$ has the following properties, for every
  $i\in\{0,\ldots,\lfloor (1/4)\log_rm\rfloor\}$:
  \begin{enumerate}
    \item For every node $v\in V_i$, $\Pr(\Xi(v)) \le (3/r)^i$,
         and
    \item For any hyperplane $\ell$, the number of elements in
      $\poly(V_i)$ whose interiors intersect $\ell$ is at most
      $(cr^{1-1/d})^i$.
  \end{enumerate}
  The sample partition tree $T_D$ can be constructed in $O(m\log m)$
  time plus the cost of $O(m)$ calls to the Sampling Oracle.
\end{thm}

\def\isdef{\buildrel {\rm def} \over =}
\def\PROB{\Pr}

\begin{proof}
  Condition~2 follows with certainty from \thmref{point-partition-tree}.
  In this proof we bound the probability of failure for Condition~1.
  Let $k=\lfloor (1/4)\log_rm\rfloor$.

  We will use $D_m(A)$ to denote the empirical measure of a set
  $A\subseteq\R^d$:
  \[ 
    D_m (A) \isdef {{|S\cap A|} \over m} \enspace .
  \]  
  From \thmref{point-partition-tree} we have
  \[
    \sup_{v\in V_i} D_m(\Xi(v)) \le (2/r)^i \enspace .
  \]
  Now,
  \begin{eqnarray*}
    \PROB \left( \sup_{v\in V_i} D ( \Xi(v) ) > \left(\frac{3}{r}\right)^i \right)
    &=& \PROB \left( \cup_{v\in V_i} \left[ D ( \Xi(v) ) - D_m(\Xi(v))  >  \left(\frac{3}{r}\right)^i - D_m(\Xi(v)) \right] \right) \cr
    &\le& \PROB \left( \cup_{v\in V_i} \left[ D ( \Xi(v) ) - D_m(\Xi(v))  > \left(\frac{3}{r}\right)^i - \left(\frac{2}{r}\right)^i  \right] \right) \cr
    &\le& \PROB \left( \sup_{v\in V_i} \left(  D ( \Xi(v) ) - D_m(\Xi(v)) \right) > r^{-i} \right) \cr
    &\le& \PROB \left( \sup_{A \in {\mathcal{A}}} \left( D ( A ) - D_m ( A ) \right)  > r^{-i} \right) \cr
  \end{eqnarray*}
  where $\mathcal{A}$ are sets formed by taking the intersection of $k$
  closed simplices and subtracting $k(r-1)$  (possibly empty) simplices
  from this intersection.  This is because $\Xi(V)\subseteq \mathcal{A}$.
  We can actually handle all levels $i\in \{0,\ldots,k\}$ of the tree
  at once with the inequality
  \begin{eqnarray*}
  \PROB \left( \cup_{i\in\{1,\ldots,k\}}\left[\sup_{v\in V_i} D ( \Xi(v) ) > \left(\frac{3}{r}\right)^i \right] \right )
  \leq \Pr \left( \sup_{A \in {\mathcal{A}}} \left( D ( A ) - D_m ( A ) \right)  > r^{-k} \right) \enspace .
  \end{eqnarray*}
  The class $\mathcal{A}$ for $k=1$ and $r=1$ is the class of all
  simplices in $\R^d$.   Since a simplex in $\R^d$ is the common
  intersection of $d+1$ halfspaces in $\R^d$, it helps to first consider
  halfspaces.  The set of halfspaces in $\R^d$ has Vapnik-Chervonenkis
  dimension $d+1$ (this follows from Radon's Theorem \cite{e93}).
  By Sauer's lemma \cite{s72}\cite[pages~28--29]{dl01}, the number of
  subsets of an $m$-point set that can be obtained by intersections
  with halfspaces does not exceed $(m+1)^{d+1}$.  A simple combinatorial
  argument then implies that the number of subsets of an $m$-point set
  that can be obtained by intersections with simplices does not exceed
  $(m+1)^{(d+1)^2}$.

  Assume now general $r$ and $k$. Then the number of subsets of an
  $m$-point set that can be obtained by intersections with sets from
  $\mathcal{A}$ does not exceed $(m+1)^{rk(d+1)^2}$, by the same
  combinatorial argument.  By a version of the Vapnik-Chervonenkis
  inequality \cite{vc71} shown by Devroye \cite{d82},
  \[
    \PROB \left( \sup_{A \in {\mathcal{A}}} \left| D ( A ) - D_m ( A ) \right|  \ge t \right)
    \le 4 \left( m^2 + 1 \right)^{rk(d+1)^2} \exp\left(4t+4t^2-2mt^2\right) \enspace ,
  \]
  for all $t>0$.
  Thus, 
  \begin{eqnarray*}
    \PROB \left( \cup_{i\in\{1,\ldots,k\}}\left[\sup_{v\in V_i} D ( \Xi(v) ) > \left(\frac{3}{r}\right)^i \right] \right)
    &\le& 4 \left( m^2 + 1 \right)^{rk(d+1)^2} \exp\left(4r^{-k}+4r^{-2k}-2mr^{-2k}\right) \\
    &\le& \exp\left(2+4r^{-k}+4r^{-2k}+rk(d+1)^2\ln(m^2+1)-2mr^{-2k}\right) \\
    &\le& \exp\left(10+r(d+1)^2\log_r(m)\ln(m^2+1)-2m^{1/2}\right) ~,\\
    & = & O\left(e^{-m^{1/2}}\right)
    \end{eqnarray*}
  since $k= (1/4)\log_rm$.
\end{proof}

Note that, so far, we have not considered the query problem $\mathcal{P}$
at all;  the sample partition tree $T_D$ of \thmref{prob-partition-tree}
is defined completely in terms of the probability measure $D$.
The \emph{odds-on tree} $T_{\mathcal{P},D}$ for $(\mathcal{P},D)$
is obtained in the following way:  We start by constructing a sample
partition tree $T_D$ as described in \thmref{prob-partition-tree}
using the value $m=n^{\tau}$ for some parameter $\tau > 0$ and setting
the maximum depth to $k=\lfloor{(1/4)\log_{r} m}\rfloor$.

Next, we trim some nodes of $T_{\mathcal{P},D}$. We use the
Interference Oracle to test, for each node $v$ of $T_D$, if
$\mathcal{P}(p)=\mathcal{P}(q)$ for all pairs of points $p,q\in
\poly(v)$.  If so, we remove all the subtrees rooted at the children
of $v$, we call $v$ a \emph{terminal leaf}, and we label $v$ with the
label $\ell(v)=\mathcal{P}(q)$.  Note that the Interference Oracle works
for simplices, but $\poly(v)$ is a polytope. Thus, this test requires
decomposing $\poly(v)$ into simplices and using the Interference
Oracle on each simplex. Using the \emph{bottom vertex triangulation}
\cite{c88} for this decomposition allows us to decompose $\poly(v)$ into
$O((\depth_T(v))^{\floor{d/2}})=\log^{O(1)} n$ simplices in $O(\log^{O(1)}
n)$ time.  This yields the following lemma:

\begin{lem}\lemlabel{odds-on-tree-construction}
  For any $n>0$ and any $\tau > 0$, an odds-on tree of size $m=n^\tau$,
  with maximum depth $k=\floor{(1/4)\log_{r/3} m}$, and having the
  properties of \thmref{prob-partition-tree} can be constructed in
  $O(n^\tau\log^{O(1)} n)$ time plus the cost of $O(n^\tau)$ calls
  to the Sampling Oracle and $O(n^\tau\log^{O(1)} n)$ calls to the
  Interference Oracle.
\end{lem}

Using an odds-on tree to answer a query $q\in \R^d$, is easy: We follow
the search path for $q$ in $T_{\mathcal{P},D}$ until we either reach
a terminal leaf $v$, in which case we output $\ell(v)$, or we reach a
non-terminal leaf $w$ after $O(\log n)$ steps, in which case we rely
on the Backup Oracle to report $\mathcal{P}(v)$ in $O(\log n)$ time.
The correctness of this procedure follows immediately from the definition
of terminal and non-terminal nodes.  In the next section, we analyze
the performance of odds-on trees.

\section{Analysis}
\seclabel{analysis}

In this section, our goal is to lower-bound the expected cost of any
linear decision tree that solves $\mathcal{P}$ in terms of the odds-on
tree $T_{\mathcal{P},D}$.  We accomplish this by decomposing the nodes
of $T_{\mathcal{P},D}$ into subsets with some helpful combinatorial
properties.

An \emph{$i$-set} of a rooted tree $T$ is a set of vertices in $T$ all of
which are at depth at most $i$ and in which no vertex in the set is the
ancestor of any other vertex in the set.  We say that a set of regions
$F=\{X_1,\ldots,X_t\}$, $X_i\subseteq\R^d$, is in \emph{$k$-general
position} if there is no hyperplane that intersects $k$ or more elements
of $F$.

\begin{lem}\lemlabel{independent}
  Let $T_{\mathcal{P},D}$ be the odds-on tree defined in
  \secref{data-structure}, let $V$ be an $i$-set of $T_{\mathcal{P},D}$,
  and let $k>1$ be a constant.  Then $V$ contains a subset $V'\subseteq
  V$ such that the elements of $\poly(V')$ are pairwise disjoint and  in
  $k$-general position and $|V'| = \Omega(|V|/r^{i(d/k+1-1/d+\delta)})$,
  where $\delta > 0$ is a decreasing function of $r$.
\end{lem}

\begin{proof}
  We will first use the probabilistic method \cite{as08} to establish
  the existence of a (not necessarily disjoint) set $V''$ satisfying the
  size and $k$-general position requirements and then show that $V''$
  contains a large subset $V'$ whose elements are also pairwise disjoint.

  Let $V''$ be a Bernoulli sample of $V$ where each element is selected
  independently with probability $p=r^{-i(d/k+1-1/d+\delta)}$. We
  will prove that
  \[
     \Pr\left\{
        \mbox{$\poly(V'')$ is in $k$-general position 
          and $|V''| = \Omega(p|V|)$}
      \right\} > 0 \enspace .
  \]
  Consider any hyperplane $\ell$. Condition~2 of
  \thmref{prob-partition-tree} implies that $\ell$ intersects the
  interior of at most $(cr^{1-1/d})^{i}$ elements of $V$ for some
  constant $c$.  The probability that $\ell$ intersects the interior of
  $k$ or more elements of $V''$ is therefore no more than
  \[
    \binom{((cr)^{1-1/d})^{i}}{k}\cdot p^k
    \le (cr)^{i(k-k/d)}p^k  \enspace .
  \]
  For each node $v\in V$, $\poly(v)$ has $O(i^{\floor{d/2}})$
  vertices, and the number of nodes in $V$ is at most $r^{i}$.
  Therefore, the elements of $\poly(V)$ define a \emph{test set} $L$ of
  $O((i^{\floor{d/2}}r^i)^d)=O(i^{d^2/2}r^{di})$ hyperplanes such that
  $\poly(V'')$ is in $k$-general position if and only if no hyperplane in
  $L$ intersects $k$ or more elements of $\poly(V'')$. The probability
  that \emph{any} hyperplane in $L$ intersects $k$ or more elements of
  $\poly(V'')$ is therefore at most
  \[
    O(i^{d^2/2}r^{di})\cdot (cr)^{i(k-k/d)}p^k
     = O(r^{i(d+k-k/d+\delta)})p^k 
     = O(r^{\delta-k\delta)})
     = o(1) \enspace 
  \]
  for any $\delta > (k-k/d)\log_r c$.
  The above argument shows that the nodes in
  $V''$ are quite likely to be in $k$-general position. To see that
  $V''$ is sufficiently large, we simply observe that $|V''|$ is a
  $\mathrm{binomal}(|V|,p)$ random variable and therefore has median
  value at least $\lfloor{p|V|}\rfloor = \Omega(|V|/r^{i(d/k+1-1/d+\delta)})$.
  Therefore,
  \[
     \Pr\left\{
        \mbox{$\poly(V'')$ is in $k$-general position 
          and $|V''| = \Omega(|V|/r^{i(d/k+1-1/d+\delta)})$ }
      \right\} \ge 1- (o(1) + 1/2) > 0 \enspace .
  \]
  This establishes the existence of a sufficiently large set $V''$
  such that $\poly(V'')$ is in $k$-general position.

  Finally, we select $V'\subseteq V''$ so that the elements of
  $\poly(V')$ are pairwise disjoint.  To do this, imagine
  sweeping a hyperplane $\ell(t)=\{(x_0,\ldots,x_d)\in\R^d:x_0=t\}$ from
  $t=-\infty$ to $+\infty$.  Associate with each element $v\in V'$, the
  maximal interval $[a_v,b_v]$ such that $\poly(v)$ intersects $\ell(t)$
  for all $t\in[a_v,b_v]$. This yields a set $S_{V''}$ of real intervals
  such that no point is contained in $k$ or more elements of $S_{V''}$.
  Dilworth's Theorem \cite{d50} implies that $S_{V''}$ contains a subset
  of size at least $|V''|/k$ of non-overlapping intervals. This subset
  corresponds to a subset $V'\subseteq V''$ with $|V'|\ge |V''|/k$
  and such that the elements of $\poly(V')$ are pairwise
  disjoint. The set $V'$ satisfies all the conditions of the lemma.
\end{proof}

We are now ready to show that the expected search time in the odds-on
tree is a lower bound on the expected cost of any linear decision tree
that solves $\mathcal{P}$.

\begin{lem}\lemlabel{lower-bound}
  Let $T_{\mathcal{P},D}$ be the odds-on tree defined in
  \secref{data-structure} and assume $T_{\mathcal{P},D}$ satisfies
  Conditions~1 and 2 of \thmref{prob-partition-tree}.  Let $L$ denote
  the set of leaves of $T_{\mathcal{P},D}$, and let $T^*$ be any linear
  decision tree that solves $\mathcal{P}$.  Then
  \[ \mu_D(T^*) = \Omega(H(\Xi(L))-1) \enspace . \]
\end{lem}

\begin{proof}
  This proof mixes the ideas from the proofs of Lemma~3 by Dujmovi\'c
  \etal\ \cite{dhm09} and Lemma~4 by Collette \etal\ \cite{cdilm08}.
  For an $i$-set $V$ of nodes in $T_{\mathcal{P},D}$, we define the
  shorthands $\Pr(V)=\Pr(\cup\Xi(V))$ and $H(V)=H(\Xi(V))$.  Let $T'$
  be the tree obtained from $T_{\mathcal{P},D}$ by removing all terminal
  leaves, and let $L'$ denote the set of leaves of $T'$.  Note that
  \[  
     H(L') = H(L) - O(\log r) 
     = H(L) - O(1)
  \]
  since each leaf in $L'$ has at most $r=O(1)$ children in $L$. 

  We first partition $L'$ into groups $G_1,G_2,\ldots$, where $G_i$
  contains all leaves $v$ such that $1/2^{i-1}\ge \Pr(\Xi(v)) \ge
  1/2^{i}$.  Note that Condition~1 of \thmref{prob-partition-tree}
  implies that the depth of a node in $G_i$ is at most $i/\log(r/3)$.

  We then further partition each group $G_i$ into subgroups
  $G_{i,1},\ldots,G_{i,t_i}$. For each $j\in\{0,\ldots,G_{i,t_i-1}\}$,
  $|G_{i,j}|= \Omega(2^{\alpha i})$, for some constant $\alpha >0$
  and the elements of $\poly(G_{i,j})$ are pairwise disjoint and in
  $k$-general position.  Furthermore, the final subgroup, $G_{i,t_i}$
  has size at most $O(2^{\beta i})$, for some constant $\beta < 1$.

  This partitioning is accomplished by repeatedly applying
  \lemref{independent} to remove a subset $G_{i,j}\subseteq G_{i}$
  that is in $k$-general position and has size $\Omega(2^{\alpha i})$,
  stopping the process once the size of $G_i$ drops below $O(2^{\beta
  i})$. This works provided that we choose $\beta$, $k$, and $r$
  so that $\beta > ((\log r)/(\log (r/3) ))((d/k+1-1/d+\delta)$ and
  set $\alpha\le\beta - ((\log r)/(\log (r/3)))(d/k+1-1/d+\delta)$.
  For example, by choosing $r$ and $k$ to be sufficiently large,
  $\beta=1-1/(3d)$ and $\alpha=1/(3d)$.

  Consider the linear decision tree $T^*$ that solves $\mathcal{P}$.
  The leaves of $T^*$ partition $\R^d$ into cells whose closures are
  convex polytopes.  For a leaf $w$ of $T^*$ we denote its polytope
  by $\poly(w)$. If the depth of $w$ is $i$, then $\poly(w)$ is
  the intersection of  at most $i$ halfspaces.  This implies that
  $\poly(w)$ intersects at most $ik$ elements of $\poly(G_{i,j})$ since,
  otherwise, $\poly(w)$ contains $\poly(v)$ for some non-terminal node
  $v\in G_{i,j}$.  (This would contradict the assumption that $T^*$
  solves $\mathcal{P}$.)

  Let $w$ be some leaf of $T^*$ such that $\poly(w)$ intersects $t$
  elements of $\poly(G_{i,j})$.  Then, by the discussion in the previous
  paragraph, $\depth_{T^*}(w) \ge \ceil{t/k}$.  We can easily create a
  subtree $T_w^*$ of $w$ whose height is at most $t$ and with the property
  that $\poly(w')$ intersects at most one element of $\poly(G_{i,j})$
  for each leaf $w'$ of $T_w^*$.\footnote{For example, making the root of
  $T^*_w$ correspond to the bisector of two polyhedra of $\poly(G_{i,j})$
  that intersect $\poly(w)$ means that each of the children of the root
  intersect at most $t-1$ elements of $\poly(G_{i,j})$. Applying this
  recursively yields a subtree of height at most $t-1$.  A slightly more
  involved argument can produce a tree $T^*_w$ of depth $O(\log t)$.} If
  we do this for every leaf $w$ of $T^*$ we obtain a tree $T^{*}_{i,j}$
  such that every leaf of $T^*_{i,j}$ intersects at most one element
  of $\poly(G_{i,j})$.

  Let $D_{i,j}=D_{|\cup\Xi(G_{i,j})}$ denote the distribution
  $D$ conditioned on $\cup\Xi(G_{i,j})$.  Note that the leaves of
  $T^*_{i,j}$ could be relabeled so that they indicate which element of
  $\poly(G_{i,j})$ (if any) that they intersect.  By Shannon's Theorem
  (\thmref{shannon}), this implies that
  \[
    (k+1)\cdot\mu_{D_{i,j}}(T^*)
       \ge \mu_{D_{i,j}}(T^*_{i,j}) 
       \ge H(G_{i,j}) \enspace .
  \]
  It follows \cite[Lemma~3]{cdilm09} that
  \[
    (k+1)\cdot\mu_D(T^*) \ge H(L') 
       - H(\{\cup \Xi(G_{i,j}):i\in\N,\, j \in\{1,\ldots,t_{i}\}) 
       - O(1) \enspace .
  \]
  Thus, all that remains is to upper-bound the contribution of
  $\bar{H}=H(\{\cup \Xi(G_{i,j}):i\in\N,\, j \in\{1,\ldots,t_{i}\})$,
  as follows:
  \begin{eqnarray*}
    \bar{H} &= & H(\{\cup G_{i,j}:i\in\N,\, j \in\{1,\ldots,t_{i}\}) \\
     & = & \sum_{i=1}^\infty \sum_{j=1}^{t_{i}} 
         \Pr(G_{i,j})\log(1/\Pr(G_{i,j})) \\
   & = & \sum_{i=1}^\infty
        \left( 
          \sum_{j=1}^{t_{i}-1} 
             \Pr(G_{i,j})\log(1/\Pr(G_{i,j}) 
             + \Pr(G_{i,t_i})\log(1/\Pr(G_{i,t_i}))
        \right) \\
   & \le & \sum_{i=1}^\infty
        \left( 
          \sum_{j=1}^{t_{i}-1} 
             \Pr(G_{i,j})\log(2^{i-\alpha i})
             + i 2^{\beta i - i + 1} + O(1)
        \right) \\
    & \le & (1-\alpha)H(L') + O(1) \enspace .
  \end{eqnarray*}
  Thus, we have 
  \begin{eqnarray*}  
     (k+1)\mu_D(T^*) 
       &\ge& H(L') - \bar{H} -O(1)  \\
       &\ge& \alpha H(L') - O(1) \\
       &\ge& \alpha H(L) - O(1) \\
       & = & \Omega(H(L) - 1) \enspace ,
  \end{eqnarray*}
  so $\mu_D(T^*) = \Omega(H(L) - 1)$, as required.
\end{proof}

\noindent\textbf{Remark:}  By more carefully handling the constants
in the proof of \lemref{lower-bound}, and allowing $k$ and $r$ to be
arbitrarily large, one can obtain the tighter lower-bound
\[
   \mu_D(T^*) + 
   \log (\mu_D(T^*)) \ge \alpha H - O(\log (kr))
\]
where $\alpha$ can be made arbitrarily close to $1/d$ by increasing $k$
and $r$.

\begin{thm}\thmlabel{odds-on}
  Let $\mathcal{P}:\R^d\rightarrow A$ be a decision problem for which we
  have a ($O(\log n)$ time) Backup Oracle and an Interference Oracle,
  and let $D$ be any probability measure over $\R^d$ for which we have
  a Sampling Oracle.  Then, for any constant $\epsilon > 0$,  an odds-on
  tree of size $O(n^\epsilon)$ can be constructed in $O(n^\epsilon)$ time
  plus the cost of $O(n^\epsilon)$ calls to the Sampling Oracle and
  $O(n^\epsilon)$ calls to the Interference Oracle.

  This odds-on tree can, in conjunction with the Backup Oracle, compute
  $\mathcal{P}(q)$ for any $q\in\R^d$ drawn according to $D$ in $O(H^*+1)$
  expected time, where $H^* \le \mu_D(T^*)$ for any linear decision tree
  $T^*$ that solves $\mathcal{P}$.
\end{thm}

\begin{proof}
  Applying \lemref{odds-on-tree-construction} with $\tau < \epsilon$
  yields the stated bounds on the construction time and the use of the
  Sampling and Interference Oracles.

  If $T_{\mathcal{P},D}$ satisfies Conditions~1 and 2 of
  \thmref{prob-partition-tree} then, by \lemref{lower-bound}, the expected
  time to answer queries using $T_{\mathcal{P},D}$ is
  \[
     \sum_{t\in L} \Pr(t)O(\depth_T(t)) 
          = \sum_{t\in L}\Pr(t)O(\log(1/\Pr(t))) = O(H(L)) \enspace .
  \]
  Otherwise, the expected time to answer queries using $T_{\mathcal{P},D}$
  is $O(\log n)$.  Therefore, the expected time (where the expectation
  is taken over $D$ and the random sampling used to generate $S$ in
  \thmref{prob-partition-tree}) to answer a query using an odds-on tree is
  \[
     (1-O(e^{-n^{\tau/2}}))\cdot O(H(L)) + O(e^{-n^{\tau/2}})\cdot O(\log n) =
        O(H(L)+1)  \enspace .
  \]
  On the other hand, by \lemref{lower-bound} the expected cost of any
  linear decision tree $T^*$ that solves $\mathcal{P}$ is
  \[
      \mu_D(T^*) = \Omega(H(L) - 1) \enspace ,
  \]
  which completes the proof.
\end{proof}

Using a standard lifting of query points \cite{yy85}, any degree $p$
algebraic decision tree that solves a problem $\mathcal{P}:\R^d\rightarrow
A$ can be implemented as a linear decision tree in $\R^{d'}$, with
$d'={p+d\choose p}-1$.  Applying this yields the following corollary:

\begin{cor}\corlabel{odds-on-algebraic}
  Let $\mathcal{P}:\R^d\rightarrow A$ be a decision problem for which we
  have an ($O(\log n)$ time) Backup Oracle and an Interference Oracle,
  and let $D$ be any probability measure over $\R^d$ for which we have a
  Sampling Oracle.  Then, for any constants $\epsilon > 0$ and $p\ge 1$,
  using $O(n^\epsilon)$ calls to the Sampling Oracle and $O(n^\epsilon)$
  calls to the Interference Oracle, an odds-on tree can be constructed in
  $O(n^\epsilon)$ time and $O(n^\epsilon)$ space.  This odds-on tree can,
  in conjunction with the backup oracle, answer $\mathcal{P}$-queries
  drawn according to $D$ in $O(H^*+1)$ expected time, where $H^* \le
  \mu_D(T^*)$ for any degree $p$ algebraic decision tree $T^*$ that
  solves $\mathcal{P}$.
\end{cor}

\section{Applications}
\seclabel{applications}

In this section, we discuss a few of the many potential applications of
odds-on trees.  In all of our applications, the problem $\mathcal{P}$ is a
query problem over some set of $n$ geometric objects in $\R^d$.  In order
to shorten the statements of the theorems in this section, we say that a
data structure for a problem $\mathcal{P}$ is \emph{distribution-sensitive
(in the linear decision tree model)} if the expected query time of
the data structure is $O(\mu_D(T^*)+1)$ for any linear decision tree
$T^*$ that solves $\mathcal{P}$. Before delving into the details of the
applications, we first outline some general strategies for implementing
the Interference Oracle needed to build an odds-on tree.

An odds-on tree can be made to have size $O(n^\epsilon)$ for any
constant $\epsilon > 0$.  Furthermore, the number of calls to the
Interference Oracle made during the construction of an odds-on tree is
$O(n^\epsilon)$.  In almost all of our applications, the Interference
Oracle can be trivially implemented to run in $O(n)$ time by testing
the query simplex $\Delta$ against each input element.  This means
that, even with no preprocessing, the contribution of calls to the
Interference Oracle to the construction time of an odds-on tree is no more
than $O(n^{1+\epsilon})$.  Some of the subsequent theorems in this paper
will use this fact implicitly.

In other cases, the Interference Oracle corresponds to a natural query for
which there exists (or we can develop) an $O(n^{1-\epsilon})$ query-time
$O(n\log n)$ preprocessing-time data structure.  In these cases, the time
to construct the odds-on tree becomes $O(n\log n)$.  One particularly
common instance of this occurs when Interference Oracle queries can
be reduced to $O(1)$ simplex range counting queries in $\R^{d'}$ for
some constant dimension $d'$.  In this case, Matou\v{s}ek's partition
trees \cite{m92} yield an Interference Oracle that can be constructed in
$O(n\log n)$ time and that can answer queries in $O(n^{1-1/d+\epsilon})$
time. 

\subsection{Point Location Problems}

The \emph{planar point location problem} is to determine which face
of a planar straight line graph $G$ contains a query point $q\in\R^2$.
A number of authors have considered distribution-sensitive algorithms for
this problem \cite{acmr00,amm00,amm01a,amm01b,ammw07,cdilm08,i01,i04}
and have mostly solved it.  The most general such result is due to
Collette \etal\ \cite{cdilm08} and gives a distribution-sensitive
data structure for point location in connected planar subdivisions.
This leaves open the case where the graph, $G$, of the subdivision is not
connected.  Since there are several $O(\log n)$ query-time, $O(n\log n)$
preprocessing-time data structures for planar point location in (possibly
disconnected) planar subdivision \cite{as98,egs86,k83,m90,st86} we can
apply odds-on trees.

To obtain a fast preprocessing time, we can use an
implementation of the Interference Oracle based on partition trees.  For
this problem, the Interference Oracle must answer queries of the form:
``Does the interior of query triangle $\Delta$ intersect more than one face
of $G$?''  We can build a data structure for this problem by first removing
from $G$ any edges and vertices not on the boundary of more than 1 face to
obtain a graph $G'$.  The interior of $\Delta$ intersects more than one
face of $G$ if and only if the interior of $\Delta$ intersects an edge of
$G'$.

If $\Delta$ intersects an edge of $G'$ then $\Delta$ contains a vertex of
$G'$ or some edge of $\Delta$ intersects some edge of $G'$.  Determining if
$\Delta$ contains a vertex of $G'$ is the classic 2-dimensional
simplex-range counting problem that can be solved in $O(n^{1/2+\epsilon})$
time after $O(n\log n)$ preprocessing using partition trees.  To determine
if some edge $uw$ of $\Delta$ intersects some edge $xy$ of $G'$, we observe
that $uw$ and $xy$ intersect if and only if
\[ L(u,w,y) \wedge L(x,y,u) \wedge L(w,u,x) \wedge L(y,x,w) \]
or
\[ R(u,w,y) \wedge R(x,y,u) \wedge R(w,u,x) \wedge R(y,x,w) \]
where $L(a,b,c)$ and $R(a,b,c)$, are predicates that are true if and
only if the sequence of points $abc$ form a left turn, respectively, a
right turn.  If we fix $x$ and $y$, then each of the above predicates is
the sign of a linear function over the variables $u_1$, $u_2$, $w_1$,
and $w_2$.  Therefore, we can treat each edge of $G'$ as a point in
$\R^4$, yielding a point set $S_{G'} \subset\R^4$ such that testing if
edge $uw$ intersects some edge of $G'$ can be reduced to two simplex
range counting queries over $S_{G'}$.  Again, partition trees allow us to
do this in $O(n^{3/4+\epsilon})$ time after $O(n\log n)$ preprocessing.
Therefore, the Interference Oracle for point location can be implemented
to run in $O(n^{3/4+\epsilon})$ time after $O(n\log n)$ preprocessing.
This yields our first theorem:

\begin{thm}\thmlabel{planar-point-location}\thmlabel{first-app}\thmlabel{a}
  There exists a distribution-sensitive data structure for the planar
  point location problem that uses $O(n\log n)$ preprocessing time and
  $O(n)$ space.
\end{thm}

The \emph{3-d point in polytope problem} is the problem of determining if
a query point $q\in\R^3$ is contained in a 3-dimensional polytope $P$.
The Dobkin-Kirkpatrick hierarchy \cite{dk83} gives an $O(\log n)$
query-time, $O(n)$ preprocessing-time data structure for this problem.
Furthermore, Interference Oracle queries (which involve testing if a
tetrahedron intersects the boundary of $P$) can also be answered in $O(\log
n)$ time by the Dobkin-Kirkpatrick hierarchy.

\begin{thm}\thmlabel{b}
  There exists a distribution-sensitive data structure for the
  3-dimensional point in polytope problem that uses $O(n)$ preprocessing
  time and space.
\end{thm}

These results can be extended to higher dimensions, though with more
space \cite{c88}:

\begin{thm}
  There exists a distribution-sensitive data structure for
  the $d$-dimensional point in polytope problem that uses
  $O(n^{\floor{d/2}+\epsilon})$ preprocessing time and space.
\end{thm}

The problem of \emph{point-location in an arrangement of hyperplanes} is
the problem of determining which cell in an arrangement of $n$ hyperplanes
in $\R^d$ contains a query point $q\in\R^d$.  Liu gives an $O(n^d)$
space and preprocessing-time, $O(\log n)$ query-time data structure for
this problem \cite{l04}.

\begin{thm}
  There exists a distribution-sensitive data structure for point location
  in an arrangement of hyperplanes that uses $O(n^d)$ preprocessing time
  and space.
\end{thm}

\subsection{Post-Office Queries}

For an $n$ point set $S\subset\R^d$, the \emph{$d$-dimensional post-office
problem} asks for a point of $p\in S$ that minimizes the Euclidean
distance $\|pq\|$ for a query point $q\in\R^d$.

The post-office problem can be solved efficiently through the use of
point location in Voronoi diagrams.  In 2-dimensions, Voronoi diagrams
are planar graphs of size $O(n)$ and can be computed in $O(n\log n)$ time
\cite{ps85}.  Combining this with \thmref{planar-point-location} gives
a distribution-sensitive data structure for the 2-d post-office problem:

\begin{thm}\thmlabel{c}
  There exists a distribution-sensitive data structure for the
  2-dimensional post office problem that uses $O(n\log n)$ preprocessing
  time and $O(n)$ space.
\end{thm}

In $d >2$ dimensions, the post-office problem can still be solved using
Voronoi diagrams, but the space and preprocessing costs are higher
\cite{c88postoffice}:

\begin{thm}
  There exists a distribution-sensitive data structure for the
  2-dimensional post office problem that uses $O(n^{\ceil{d/2}+\epsilon})$
  preprocessing time and space.
\end{thm}

\subsection{Ray Shooting}

The \emph{ray shooting in a polygon problem} asks for the first point
on the boundary of a polygon $P$ intersected by a query ray.  A query
ray can be represented as a pair of points in $\R^2$ (the source and
any other point on the ray), so this is a query problem over $\R^4$.
Several $O(n)$ preprocessing time $O(\log n)$ query time solutions to
this problem exist \cite{cegghss94,hs95}.

\begin{thm}
  There exists a distribution-sensitive data structure for the ray shooting
  in a polygon problem that uses $O(n^{1+\epsilon})$ preprocessing time and
  $O(n)$ space.
\end{thm}

The \emph{ray shooting in a $d$-dimensional polytope problem}
asks for the first point on the boundary of a convex polytope
$P\subseteq\R^d$ intersected by a query ray.  For polytopes in $\R^3$,
the Dobkin-Kirkpatrick hierarchy \cite{dk83} can perform ray shooting
in $O(\log n)$ time per query.

\begin{thm}
  There exists a distribution-sensitive data structure for the
  ray shooting in a $3$-dimensional polytope problem that uses
  $O(n^{1+\epsilon})$ preprocessing time and $O(n)$ space.
\end{thm}

Schwarzkopf \cite{s92} gives an $O(n^{\floor{d/2}+\epsilon})$ space,
$O(\log n)$ query time solution for the problem of ray shooting in a
$d$-dimensional polytope with $d \ge 4$.

\begin{thm}
  There exists a distribution-sensitive data structure for the
  ray shooting in a $d$-dimensional polytope problem that uses
  $O(n^{\floor{d/2}+\epsilon})$ preprocessing time and space.
\end{thm}

\subsection{Orthogonal Range Counting}

The \emph{2-d orthogonal range counting problem} is the problem
of counting the number of points of a data set $S\subset\R^2$ that
are contained in a query rectangle $[q_1,q_2]\times[q_3,q_4]$. (Note
that this produces a query point $q\in \R^4$.) Bentley's range trees
\cite{b75}, with fractional cascading \cite{cg86,l78}, yield an $O(n\log
n)$ preprocessing time and space, $O(\log n)$ query time data structure
for this problem.  Interference Oracle queries for 2-d orthogonal
range counting require determining if any rectangle represented by a
point in a 4-dimensional simplex has some point of $S$ on its boundary.
This problem can be decomposed into $O(1)$ simplex range counting queries
in $\R^4$ in a manner similar to that used for the point location problem.
Thus, an Interference Oracle for this problem can be implemented in
$O(n^{3/4+\epsilon})$ time after $O(n\log n)$ preprocessing.

\begin{thm}\thmlabel{d}
  There exists a distribution-sensitive data structure for the 2-d
  orthogonal range counting problem that uses $O(n\log n)$ preprocessing
  time and $O(n\log n)$ space.
\end{thm}

Range trees satisfy the constraints of the comparison tree model of
computation, which is considerably weaker than the linear decision tree
model. Dujmovi\'c \etal\ \cite{dhm09} give a distribution-sensitive 2-d
orthogonal range counting data structure that works in the comparison
tree model.  However, their technique can only answer \emph{2-sided
queries}, i.e., queries of the form $[q_1,\infty)\times[q_2,\infty)$.

Filter trees can be made to work in the comparison tree model
and handle full (4-sided) 2-d orthogonal range counting queries.
The key modification required is the use of $k$-d trees in
\thmref{prob-partition-tree} rather than Matou\v{s}ek's partition
trees. (This method, in 2 dimensions, is essentially how Dujmovi\'c \etal\
obtain their results.)

In the comparison tree model, an Interference Oracle query is a
$d$-dimensional box, rather than a simplex. In the special case of 2-d
orthogonal range counting, Interference Oracle queries can be reduced to
a constant number of rectangular range counting queries among the input
set $S$.  Therefore, Interference Oracle queries can, in this case,
be answered in $O(\log n)$ time after $O(n\log n)$ preprocessing using
range trees.

\begin{thm}\thmlabel{last-app}
  There exists a distribution-sensitive data structure \emph{in the
  comparison tree model} for the 2-d orthogonal range counting problem
  that uses $O(n\log n)$ preprocessing time and space.
\end{thm}

\section{Summary and Conclusions}
\seclabel{conclusions}

We have presented a data structure --- the odds-on tree --- that can
be added on to any data structure that answers queries drawn from some
distribution $D$ over $\R^d$.  If the underlying data structure has query
time $O(\log n)$, then the combined data structure will have optimal
expected query time in the linear decision tree model.  A variant of
the odds-on tree based on $k$-d trees provides a similar result in the
comparison tree model.  \secref{applications} contains a smattering of
applications of the odds-on tree.  Many more are possible.

Note that \corref{odds-on-algebraic} applies to all of the problems
in Theorems~\ref{thm:first-app}--\ref{thm:last-app} to yield, for
any constant $p$, data structures that are distribution-sensitive
in the degree $p$ algebraic decision tree model.  However, when
applying \corref{odds-on-algebraic}, the Interference Oracle becomes
considerably more complicated, so that the fast preprocessing times in
Theorems~\ref{thm:a}, \ref{thm:b}, \ref{thm:c}, and \ref{thm:d} increase
to $O(n^{1+\epsilon})$.
However, the real power of \corref{odds-on-algebraic} is its applications
to problems that cannot be solved by finite linear decision trees.
Examples of such problems include fixed-radius circular ray shooting in
polygons \cite{cceo04}, point-location in 2-dimensional power diagrams
\cite{a87}, point-location in planar subdivisions whose edges are defined
by algebraic curves (such as arrangements of circles), and many others.

We believe that the odds-on tree may actually be practical in some
settings.  The comparison-tree version of the odds-on tree is based on
$k$-d trees, which are simple and widely used in practice. At worst,
a comparison-based odds-on tree will perform $O(\epsilon\log n)$ extra
comparisons before falling back to the backup data structure.

In low dimensions, simpler alternatives to Matou\v{s}ek's Partition
Theorem (\thmref{point-partition}) are available and may be more
practical.  For example, in 2-d one can, $O(n)$ time, find two lines
that partition any $m$ point set into four point sets each of size at
most $\ceil{m/4}$ and such that no line intersects more than the 3 of
the resulting sets \cite{m85}. This result is strong enough that it can
be used in place of \thmref{point-partition} to prove all our results
(for 2-d problems) and yields an odds-on tree that only requires 2
point-line comparisons at each node.

Because the odds-on tree is so small (of size $O(n^\epsilon)$) it
may be useful to speed up searching in environments where memory is
constrained. For example, it can be applied to the succinct point
location data structures of Bose \etal\ \cite{bchmm09} to obtain a
distribution-sensitive and succinct data structure for point location.
In external-memory settings, an odds-on tree may also be useful as a
filter that is sufficiently small to fit into internal memory while the
backup structure lives in (larger but slower) external memory.

\bibliographystyle{plain}
\bibliography{odds-on}
\end{document}